\documentclass[11pt]{article}

\usepackage[margin=1in]{geometry}

\usepackage{hyperref}
\usepackage{appendix}

\usepackage{graphicx,epsfig,amsmath,latexsym,amssymb,verbatim}

\newcommand{\extra}[1]{}

\usepackage{amsmath,amssymb,amsfonts}
\usepackage{amsthm}

\newtheorem{theorem}{Theorem}[section]

\newtheorem{lemma}[theorem]{Lemma}

\theoremstyle{remark}

\def\squareforqed{\hbox{\rlap{$\sqcap$}$\sqcup$}}
\def\qed{\ifmmode\squareforqed\else{\unskip\nobreak\hfil
\penalty50\hskip1em\null\nobreak\hfil\squareforqed
\parfillskip=0pt\finalhyphendemerits=0\endgraf}\fi}
\def\endenv{\ifmmode\;\else{\unskip\nobreak\hfil
\penalty50\hskip1em\null\nobreak\hfil\;
\parfillskip=0pt\finalhyphendemerits=0\endgraf}\fi}

\renewenvironment{proof}{\noindent \textbf{{Proof~} }}{\qed\medskip}
\newenvironment{proof+}[1]{\noindent \textbf{{Proof #1~} }}{\qed\medskip}

\mathchardef\ordinarycolon\mathcode`\:
\mathcode`\:=\string"8000
\def\vcentcolon{\mathrel{\mathop\ordinarycolon}}
\begingroup \catcode`\:=\active
  \lowercase{\endgroup
  \let :\vcentcolon
  }

\newcommand{\nc}{\newcommand}
\nc{\rnc}{\renewcommand}
\nc{\beq}{\begin{equation}}
\nc{\eeq}{{\end{equation}}}
\nc{\beqa}{\begin{eqnarray}}
\nc{\eeqa}{\end{eqnarray}}
\nc{\lbar}[1]{\overline{#1}}
\nc{\bra}[1]{\langle#1|}
\nc{\ket}[1]{|#1\rangle}
\nc{\ketbra}[2]{|#1\rangle\!\langle#2|}
\nc{\braket}[2]{\langle#1|#2\rangle}
\nc{\proj}[1]{| #1\rangle\!\langle #1 |}
\nc{\avg}[1]{\langle#1\rangle}
\nc{\smfrac}[2]{\mbox{$\frac{#1}{#2}$}}
\nc{\tr}{\operatorname{tr}}
\nc{\tracedist}[1]{\Delta_{}\!\left( #1 \right)}
\nc{\fid}[1]{F\!\left( #1 \right)}

\nc{\ox}{\otimes}
\nc{\dg}{\dagger}
\nc{\dn}{\downarrow}
\nc{\cA}{{\cal A}}
\nc{\cB}{{\cal B}}
\nc{\cC}{{\cal C}}
\nc{\cD}{{\cal D}}
\nc{\cE}{{\mathcal E}}
\nc{\cF}{{\cal F}}
\nc{\cG}{{\cal G}}
\nc{\cH}{{\cal H}}
\nc{\cI}{{\cal I}}
\nc{\cJ}{{\cal J}}
\nc{\cK}{{\cal K}}
\nc{\cL}{{\cal L}}
\nc{\cM}{{\cal M}}
\nc{\cN}{{\cal N}}
\nc{\cO}{{\cal O}}
\nc{\cP}{{\cal P}}
\nc{\cR}{{\cal R}}
\nc{\cS}{{\cal S}}
\nc{\cT}{{\cal T}}
\nc{\cU}{{\cal U}}
\nc{\cX}{{\cal X}}
\nc{\cY}{{\cal Y}}
\nc{\cZ}{{\cal Z}}

\nc{\entI}{{\bf I}}
\nc{\entIarrow}{{\bf I}^{\leftarrow}}
\nc{\entH}{{\bf H}}
\nc{\entS}{{\bf S}}
\nc{\entHmin}{\mathbf{H}_{\min}}
\nc{\entHmax}{\mathbf{H}_{\max}}

\nc{\binent}{h_2}

\nc{\entF}{{\bf E}_f}

\nc{\isom}{\simeq}

\nc{\rank}{\operatorname{rank}}
\nc{\rar}{\rightarrow}
\nc{\lrar}{\longrightarrow}
\nc{\polylog}{\operatorname{polylog}}
\nc{\poly}{\operatorname{poly}}
\nc{\1}{\mathbb{I}}

\nc{\weight}{\textbf{w}}
\nc{\hamdist}{d_{H}}

\def\p{\pi}

\nc{\Sp}{{{\mathbb S}}}
\nc{\RR}{{{\mathbb R}}}
\nc{\CC}{{{\mathbb C}}}
\nc{\FF}{{{\mathbb F}}}
\nc{\NN}{{{\mathbb N}}}
\nc{\ZZ}{{{\mathbb Z}}}
\nc{\PP}{{{\mathbb P}}}
\nc{\QQ}{{{\mathbb Q}}}
\nc{\UU}{{{\mathbb U}}}
\nc{\OO}{{{\mathbb O}}}
\nc{\EE}{{{\mathbb E}}}
\nc{\id}{{\operatorname{id}}}

\nc{\qubitchannel}{\id_2}
\nc{\bitchannel}{\overline{\id}_2}

\nc{\be}{\begin{equation}}
\nc{\ee}{{\end{equation}}}
\nc{\bea}{\begin{eqnarray}}
\nc{\eea}{\end{eqnarray}}
\nc{\<}{\langle}
\rnc{\>}{\rangle}
\nc{\Hom}[2]{\mbox{Hom}(\CC^{#1},\CC^{#2})}
\nc{\rU}{\mbox{U}}

\nc{\ob}[1]{#1}

\nc{\unif}{\textrm{unif}}

\nc{\inter}{\textrm{int}}
\nc{\ed}{\textrm{ed}}

\nc{\grade}{\mathsf{G}}

\nc{\pguess}{P_{guess}}

\nc{\barA}{\overline{A}}
\nc{\barB}{\overline{B}}
\nc{\barC}{\overline{C}}
\nc{\barD}{\overline{D}}
\nc{\barR}{\overline{R}}
\nc{\barX}{\overline{X}}
\nc{\barY}{\overline{Y}}
\nc{\barU}{\overline{U}}

\nc{\barrho}{\overline{\rho}}

\nc{\barp}{\overline{p}}

\nc{\Pos}{\mathrm{Pos}}

\def\cK{{\mathcal K}}
\def\cU{{\mathcal U}}

\newcommand{\real}{\mathbb{R}}

\newcommand{\gor}{\rightarrow}

\newcommand{\mrg}{\mathrm{marg}}
\newcommand\numberthis{\addtocounter{equation}{1}\tag{\theequation}}

\newcommand{\Supp}{\textrm{Supp}}

\begin{document}


\title{Computational Aspects of Private Bayesian Persuasion} 
\author{Yakov Babichenko\thanks{Technion - Israel Institute of Technology. {\tt yakovbab@tx.technion.ac.il} } \and Siddharth Barman\thanks{Indian Institute of Science. {\tt barman@csa.iisc.ernet.in}}}
\date{}

\maketitle

\begin{abstract}
We study computational questions in a game-theoretic model that, in particular, aims to capture advertising/persuasion applications such as viral marketing. Specifically, we consider a multi-agent Bayesian persuasion model where an informed sender (marketer) tries to persuade a group of agents (consumers) to adopt a certain product. The quality of the product is known to the sender, but it is unknown to the agents. The sender is allowed to commit to a signaling policy where she sends a private signal---say, a viral marketing ad---to every agent. This work studies the computation aspects of finding a signaling policy that maximizes the sender's revenue. 

We show that if the sender's utility is a submodular function of the set of agents that adopt the product, then we can efficiently find a signaling policy whose revenue is at least $(1-1/e)$ times the optimal. We also prove that for submodular utilities, approximating the sender's optimal revenue by a factor better than $(1-1/e)$ is {\rm NP}-hard and, hence, the developed approximation guarantee is essentially tight. When the sender's utility is a function of the number of agents that adopt the product (i.e., the utility function is anonymous), we show that an optimal signaling policy can be computed in polynomial time. Our results are based on an interesting connection between the Bayesian persuasion problem and the evaluation of the concave closure of a set function. 

\end{abstract}

\section{Introduction}
The advent of social networks has fundamentally changed the landscape of advertising. In recent years much attention has been devoted to new forms of advertising such as viral marketing. The goal of these marketing strategies is to advertise a product by leveraging the spread of influence between consumers in a social network. To accomplish this objective, viral marketing aims to identify ``opinion leaders'' in the network who can start a cascade of influence leading to large adoption of the advertised product. 

Motivated by such marketing applications, the notable work \cite{kk} addresses the algorithmic problem of determining opinion leaders (influential individuals) in a social network; also see~\cite{do}. In~\cite{kk} the process/dynamic via which influence spreads through a social network---i.e., the extent to which people in the network influence each other---is part of the input. More formally, given a network, a dynamic that spreads influence in favor of a product through the network, and a number $n$,~\cite{kk} address the following question: which initial set of $n$ consumers (influenced, say, by giving free samples) maximize the final adoption of the product, after the influence has been spread by the underlying dynamic? Note that this model assumes that the marketer always succeeds in persuading the initial $n$ consumers to adopt the product. 

In order to better capture viral marketing, we can complement the work of~\cite{kk} by considering the following real-world aspects. The marketer might be unsuccessful in persuading a consumer to adopt a product; furthermore, some consumers might be harder to persuade than others. Another factor that should be taken into account in this scenario is the asymmetry in information; in particular, it is reasonable to assume that the marketer has more information about product's quality than the consumer. In this paper we focus on a model that is motivated by these observations, and we ask a question complementary to the one addressed in~\cite{kk}: If a marketer tries to persuade $n$ consumers to adopt a product, what is her best strategy for doing so?
  
We consider a situation where a marketer, henceforth \emph{sender}, tries to persuade potential consumers, henceforth \emph{agents}, to adopt a certain product. The asymmetry in the information can be captured by a Bayesian model where the sender knows the state of nature (for instance, whether the quality of the product is \emph{high} or \emph{low}), and agents do not. This state of nature might effect each agent's utility. The sender is allowed to commit to a (information) revelation policy regarding the state of nature\footnote{Without the possibility to \emph{commit} to a policy, the described model will correspond to a \emph{cheap talk model} (see \cite{cs}), wherein it is known that a sender's messages are not reliable for the case at hand and, therefore, the sender has no persuasion power.\\ The assumption that sender is allowed to commit can be motivated by a \emph{repeated} interaction where agents may detect deviations of the sender from the committed strategy, see \cite{rs} for a discussion on this issue.} before the actual state of nature is realized. Agent's information, that has been received from sender, may effect her decision on whether to adopt the product. Thus, such partial revelation of information may increase sender's profit.

The above described model is called \emph{Bayesian persuasion}, and has been extensively studied in recent years in contexts such as voting~\cite{ca,sch}, regulation policies~\cite{go,tam}, and marketing~\cite{br,and}. A central thread of research in  Bayesian persuasion has been the study of the optimal revenue of the sender and the optimal policy to persuade agents. Hence, considering Bayesian persuasion with a computational lens raises the following natural questions: Can we find the sender's optimal policy in polynomial time? If not, then what  fraction of the optimal revenue can be guaranteed by policies that are efficiently computable? These questions form the basis of our work. 


Bayesian persuasion and related algorithmic questions can be considered in two broad settings. The first setting considers the setup wherein the sender announces a \emph{public} signal to the agents, see, e.g., \cite{sch,ca}. The other setting models scenarios where the sender has a private communication channel with each agent, and is allowed to send a private signal---say, a viral marketing ad---to each agent; see, e.g.,~\cite{AB,taneva,wan}. In this paper we focus on latter setting. 

\noindent 
{\bf Our Results and Techniques:} Note that in a multiagent model the sender's utility is a function of the \emph{subset} of agents that adopt the product. It has been observed that for various dynamics---via which influence spreads in a network---the utility of the sender (in the final adoption state, i.e., after the dynamic had been applied) is submodular, see~\cite{kk} and references therein. Thus, settings wherein the sender's utility is submodular is of central interest in multiagent Bayesian persuasion model. For submodular utility we develop a tight $(1-\frac{1}{e})$ approximation of the sender's optimal policy. In particular, we establish the existence of a polynomial algorithm that finds a policy whose revenue is at least $(1-\frac{1}{e})$ times the optimal. On the negative side, we show that it is {\rm NP}-hard to approximate the sender's optimal revenue by a factor better than $(1-\frac{1}{e})$. 
 
\emph{Anonymous} utilities constitute another interesting class of utility functions. Here, the sender's utility is a function of the \emph{number} of agents that adopt the product. We emphasize that agents' utility functions may different for different agents (i.e., it might be harder to persuade some agents than others). We show that for a sender with an anonymous utility function, an optimal policy and the optimal revenue can be computed in polynomial time.

Our proofs are based on an interesting connection between the Bayesian persuasion model and the \emph{concave closer of a set function}~\cite{mur,dug,vondrak2007submodularity}. In Section \ref{sec:pre}---specifically, in Lemma \ref{lem:aprox}---we show that computation (or approximation) of sender's optimal revenue with utility $V:2^{[n]} \rightarrow \mathbb{R}_+$ is (computationally) equivalent to evaluation (or approximation) of the concave closer of the function $V$, here $[n]$ is the set of agents. Concave closer has been studied in the context of submodular maximization, see, e.g., \cite{ccpv,vondrak2007submodularity}, where it is primarily used as a {technical} tool to obtain approximation results. Unlike \cite{ccpv,vondrak2007submodularity}, in the Bayesian persuasion problem the concave closer turns out to be the core object. 

We establish a tight approximation bound for computing the concave closer of monotone submodular functions. Specifically, we develop a $(1-\frac{1}{e})$-approximation algorithm for computing the concave closure of monotone submodular functions. Our approximation result rests on an approximation preserving reduction between computing the concave closure and the problem of maximizing a monotone submodular function subject to a matroid constraint. Since the latter problem admits a $(1-\frac{1}{e})$ approximation (see \cite{ccpv}) the desired result follows. For the hardness result, we use tools from~\cite{klmm} and~\cite{feige} which were developed to establish the hardness of approximating the maximum social welfare in combinatorial auctions. 

We establish the result for anonymous utility functions (which are not necessarily submodular) by developing a polynomial-time algorithm for computing the concave closer of anonymous functions. Such a polynomial-time algorithm might be of independent interest (beyond the Bayesian persuasion framework), since the concave closer is a basic notion in discrete convex analysis (see e.g.,~\cite{dca}). 

Throughout the paper we consider a model wherein there are exactly two possible states of nature ($|\Omega| = 2$); such cases capture settings where binary distinctions (e.g., the quality of the product is \emph{high} or \emph{low}) suffice. In general, though, we need to address environments where there are multiple states of nature and extending our results to handle models in which $|\Omega| > 2$ is an interesting direction for future work. 


\subsection{Additional Related Work} 
The current literature on Bayesian persuasion starts with the result of~\cite{kg} who---building upon the classical work by~\cite{am}---analyze the case of a single sender and a single receiver. Several extensions of this model appear in recent papers, see, e.g., \cite{AC,GK,GK2}. One natural extension, considers the scenario where there are multiple receivers. The setting wherein the sender is only allowed to send a public signal to the receivers is considered in \cite{ca,sch}. The complementary setting in which the sender is allowed to send private signals to the receivers has been studied in \cite{AB,taneva,wan}. Our result is most closely related to the work of~\cite{AB} where the optimal policy and the optimal revenue are characterized for supermodular utilities, supermajority utilities and submodular utilities which are also anonymous. In particular, we build upon the work of~\cite{AB} with a computation perspective. Specifically, we show that for \emph{every} anonymous utility the optimal revenue can be computed in polynomial time and provide tight approximations for general submodular utilities (not necessarily anonymous). Our inapproximability result for submodular functions (Theorem \ref{thm:hard}) indicates that a closed-form expression for sender's revenue that is derived in \cite{AB} for the submodular anonymous case (and also the supermodular case) cannot have an analogous closed formula for general submodular functions.

A recent, interesting paper by~\cite{DX} studies the complexity of the Bayesian persuasion problem in the case of single agent single receiver with $n$ actions to the receiver and $exp(n)$ states of nature. \cite{DX} prove that in the case where payoff profiles are i.i.d. distributed (for all receiver's actions) the problem can be solved in polynomial time. The same is not true if payoff profiles are independently distributed (but not identically)-- the problem becomes $\#$P-hard. Finally, they show that the general problem (with arbitrary payoff profiles) can be approximately solved efficiently in the query model, if we assume that receiver follows the recommended action by the sender in all cases where no $\varepsilon$-better action (according to his belief) exists.\footnote{This notion is called $\varepsilon$-incentive compatibility.} 

\section{Notations and Preliminaries}\label{sec:pre}

We consider a Bayesian persuasion model with a sender and $n$ agents, $[n]=\{1,2,...,n\}$. Write $\Omega=\{\omega_0,\omega_1\}$ to denote the two possible states of nature. 
Each agent $i \in [n]$ has two actions, $\{0,1\}$, and a utility function, $u_i$, that depends on the state of nature and its own action, $u_i: \Omega \times \{0,1\}\gor \real$. All agents share a common prior distribution, where $0<\gamma<1$ is the probability of state $\omega_1$, and $1-\gamma$ of state $\omega_0$. Note that even though the agents' utilities depend on the realized state of nature they are a priori unaware of it.   Throughout, we will use $\Delta(A)$ to denote the set of probability distributions over set $A$. 

\cite{AB} show that without loss of generality we can assume $u_i(\omega_0,0)>u_i(\omega_0,1)$ and $u_i(\omega_1,1)>u_i(\omega_1,0)$ for all agents $i\in [n]$. In particular, it is show in~\cite{AB} that we can efficiently reduce an instance with arbitrary utility function to an instance wherein agents prefer to adopt the product (i.e., play $1$) in state space $\omega_1$ and prefer not to adopt it (i.e., play $0$) in state $\omega_0$. Hence, throughout the paper we will work with this assumption on agents' utilities. 

As mentioned earlier, the sender's utility, $V$, depends on the set of agents that play action $1$ (i.e., the set of agents that adopt the product), $V: \{0,1\}^n\gor\real$. With a slight abuse of notation, for a subset $S \subseteq [n]$, we will use $V(S)$ to denote $V(1_S)$, where $1_S$ is the characteristic vector of subset $S$. We assume throughout that the sender's utility monotonically increases with the set of agents that play action $1$: $V(S) \leq V(T)$ for every $S\subseteq T$. 

Note that we have restricted our attention to the case wherein the sender's utility does not depend on $\Omega$. More generally, the sender's utility can be defined to be a function of the state space as well, i.e., we can have $V: \Omega \times \{0,1\}^n\gor\real$. It is shown in~\cite{AB} that such a general case can always be efficiently reduced to one in which $V(\omega_0,S)= V(\omega_1,S)$ for every $S$. Hence, throughout the paper we will focus on a utility functions, $V$, that are state independent.

As is typical is Bayesian persuasion models, we assume that only the sender knows the realized state. The agents remain unaware of it. Furthermore, following the model of Kamenica and Gentzkow \cite{kg} we allow the sender to commit in advance to an information revelation policy. In this work, however, we allow the sender to reveal the information to every agent privately. This translates to a state dependent signaling distribution. Formally, a policy of the informed sender consists of $n$ finite sets $\{\Theta_i\}_{i=1,\ldots,n}$, where $\Theta_i$ is the private signal set of agent $i$, and a mapping $F: \Omega \gor \Delta(\Theta_1\times\cdots\times \Theta_n)$. The sender can commit to a policy $F$ that is known to the agents prior to stage where the state $\omega$ is realized.  

The sequence  of the interaction between the sender and the agents is as follows. First, the sender commits to a signaling policy $F$. Then, a state $\omega\in\Omega$ is realized in accordance with the prior $(\gamma,1-\gamma)$. After that a profile of signals $\theta=(\theta_1,\ldots,\theta_n)$ is generated according to the distribution $F(\omega)$. Every agent $i$ observes her private signal realization $\theta_i\in \Theta_i$ and forms a posterior 
$\p_F(\omega_1|\theta_i)=p(\theta_i)$. 

With the posterior in hand, agent $i$ selects an action that maximizes her expected utility. In other worlds, agent $i$ plays action $1$ if and only if
\begin{align*}
p(\theta_i)u_i(\omega_1,1)+(1-p(\theta_i))u_i(\omega_0,1)\geq p(\theta_i)u_i(\omega_1,0)+(1-p(\theta_i))u_i(\omega_0,0).
\end{align*}
We assume that in case of indifference agents plays action $1$. Let $g_i(\theta_i)\in \{0,1\}$ denote agent $i$'s best-reply action when she observes the signal $\theta_i$. Also, write $g(\theta)$ to be the action profile of the agents when the realized vector of signals is $\theta$. Write $\Theta :=\Theta_1\times\cdots\times \Theta_n$. We will use $F_1\in\Delta(\Theta)$ to denote the signal distribution conditional on state $\omega_1$ and $F_0\in\Delta(\Theta)$ to denote the signal distribution conditional on state $\omega_0$. 

Let $s(F)$ be the sender's utility from the policy $(\Theta,F)$:
\begin{equation}\label{eq:rev}
s(F):=\gamma \mathbb{E}_{\theta \sim F_1} [V(g(\theta))]+(1-\gamma)\mathbb{E}_{\theta \sim F_0} [V(g(\theta))].
\end{equation}

A signaling policy $(\Theta,F)$ is said to be \emph{optimal} if it maximizes sender's utility among all possible signal sets $\Theta$ and all possible signals $F:\Omega \rightarrow \Delta(\Theta)$. 

We begin by stating a result from~\cite{AB} (see Lemma 1 in~\cite{AB}) that shows the existence of an optimal policy with the following useful properties:
\begin{itemize}
\item For every agent $i$, the private signal set $\Theta_i$ is equal to $\{0,1\}$ and agent $i$'s best reply $g_i(\theta_i)=\theta_i$. In other words, when signal $\theta_i$ is recommended by the sender to agent $i$ it is profitable (after agent $i$ performs a Bayesian update of her belief on the state of the world) for her to follow the recommendation. In~\cite{kg}, such policies are called \emph{straightforward}.
\item In the optimal policy $F_1(1,1,\ldots,1)=1$, i.e., when state $\omega_1$ is realized the sender recommends everyone to adopt the product. Recall that $F_1$ is a distribution over $\Theta$, which (by the above mentioned property) is $\{0,1\}^n$ for the optimal policy under consideration. 
\item When the realized state is $\omega_0$, the sender recommends to agent $i$ to adopt the product with probability of at most $a_i := \min(\frac{\gamma}{1-\gamma}\frac{u_i(\omega_1,1)-u_i(\omega_1,0)}{u_i(\omega_0,0)-u_i(\omega_0,1)},1)$. Write $F_0(\theta_i = 1) := \sum_{ \theta \in \{0,1\}^n: \ \theta_i = 1 } F_0(\theta)$. We succinctly express this condition as $F_0(\theta_i=1)\leq a_i$. The number $a_i$ can be interpreted as the maximal probability that the sender can ``lie'' to the agent, and will be called the \emph{persuasion level of player $i$.} 
\end{itemize}
Under such an optimal policy, sender's utility is given by
\begin{align}\label{eq:maxi0}
s(F)=\gamma V([n]) + (1-\gamma) \mathbb{E}_{\theta \sim F_0} V(\theta)
\end{align} 

It is shown in~\cite{AB} that the requirements $F_1(1,1,\ldots,1)=1$ and $F_0(\theta_i=1)\leq a_i$ in fact ensure that the best reply of every agent is equal to the sender's recommendation, $g_i(\theta_i)=\theta_i$. Also, the above mentioned properties imply that we can restrict our attention to binary signal sets, $\Theta_i =\{ 0, 1\}$ and $\Theta = \{0,1\}^n$ along with a single distribution $F_0 \in \Delta(\{0,1\}^n)$. 

Overall, in light of these properties the problem of determining an optimal policy (over general signal sets $\Theta$ and mappings $F:\Omega\rightarrow \Delta(\Theta)$)  reduces to the following well-structured maximization problem: 
\begin{align}\label{eq:maxi1}
\textrm{maximize }  \ \ & \mathbb{E}_{\theta \sim F_0} V(\theta) \ \text{ subject to } \ F_0(\theta_i=1)\leq a_i.
\end{align}

Note that prior $\gamma$ and utility $V([n])$ are fixed parameters. Hence, an optimal solution of \eqref{eq:maxi1} gives us the optimal value in \eqref{eq:maxi0}

For each subset $S\subset [n]$, with characteristic vector $1_S$, write $\mu_S$ to be the probability that exactly the agents is $S$ will receive the recommendation to adopt the product,  $\mu_S:=F_0(1_{S})$. For a given persuasion levels profile $a :=(a_1,...,a_n)$, the maximization problem \eqref{eq:maxi1} can be written as
\begin{align}\label{eq:maxi2}
V^+(a) := \ & \max \sum_{S \subseteq [n]} \mu_S V(S) \\ & \text{ s.t. } \ \ \sum_{S \subset [n]}\mu_S 1_S\leq a \nonumber \\ & \qquad \ \ \sum_{S \subset [n]}\mu_S =1 \nonumber \\ & \qquad \ \ \mu_S\geq 0. \nonumber
\end{align}

Interestingly, the expression $V^+(a)$ exactly defines the notion of \emph{concave closure} (see e.g., \cite{do}) of the sender's utility function $V$ at $a \in [0,1]^n$. We will refer to solving (approximating) the optimization problem \eqref{eq:maxi2} as computing (approximating) the concave closure. 
Note that computing the concave closer corresponds to solving a linear programming with polynomial number of constrains, but with an exponential number of variables (the variables are $\mu_S$ for every $S \subset [n]$). 

In some cases we will be interested in {approximating the optimal revenue of the sender} and, hence, we introduce here the following lemma that states that computing (approximating)  the concave closure is computationally equivalent to computing (approximating) the sender's optimal revenue. Note that, for a given parameter $\alpha \in (0,1]$, an $\alpha$ approximation of the concave closure corresponds to a distribution $\{ \mu_S\}_{S \subseteq [n]}$ that satisfies the feasibility constraints of the optimization problem \eqref{eq:maxi2} and obtains an objective function value, $\sum_{S \subseteq [n]} \mu_S V(S)$, that is at least $\alpha$ times the optimal. The next lemma states that there exists an approximation-preserving, polynomial-time reduction between computing the concave closure and finding the optimal revenue of the sender. Specifically, the lemma establishes that computing the concave closure of the sender's utility function lies at the core of determining a revenue-maximizing policy for the sender.
 
\begin{lemma}\label{lem:aprox}
Given a persuasion profile $a \in [0,1]^n$ and utility function $V$ along with an $\alpha$ approximation of the concave closer $V^+(a)$, in polynomial time we can find a policy for the sender (with utility function $V$ and persuasion profile $a$) that obtains revenue at lest $\alpha$ the optimal. 

Furthermore, for every $\varepsilon>0$, there exists a polynomial-time reduction from the problem of $\alpha$ approximating a sender's revenue (with utility function $V$ and persuasion profile $a$)  to the problem of computing an $(\alpha+\varepsilon)$ approximation of the concave closer $V^+(a)$.
\end{lemma}

\begin{proof}
The forward direction is direct: by equation \eqref{eq:maxi0}, an $\alpha$ approximation of $\max \mathbb{E}_{\theta \sim F_0} V(\theta)$ is also an $\alpha$ approximation of $\gamma V(N) + (1-\gamma) \mathbb{E}_{\theta \sim F_0} V(\theta)$.

For the other direction, given function $V$ and persuasion $a=(a_i)_{i\in [n]}$, we can set the prior $\gamma$ to be very small (e.g., $\gamma=\frac{\varepsilon (V(N))^2}{1-\alpha}$ suffices) and we set agent $i$ utilities to be
\begin{align*}
u_i(\omega_0,0)=1, \ u_i(\omega_0,1)=u_i(\omega_1,0)=0, \ u_i(\omega_1,1)=a_i \frac{1-\gamma}{\gamma}.
\end{align*}
Such a choice guarantees that indeed $a_i=\min(\frac{\gamma}{1-\gamma}\frac{u_i(\omega_1,1)-u_i(\omega_0,1)}{u_i(\omega_0,0)-u_i(\omega_0,1)},1)$. It follows that for such instances an $\alpha$ approximation of the sender's revenue implies $(\alpha+\varepsilon)$ approximation of the concave closer of $V$. 
\end{proof} 

In subsequent sections we establish algorithmic and hardness results for the problem of finding the optimal policy (and revenue) of the sender. We do so by using the above mentioned lemma and, in particular, addressing the computation of the concave closure.

\section{Anonymous Utility} \label{sect:anon}

This section considers the case wherein the sender's utility function is anonymous i.e., it satisfies $V(S)=f(|S|)$ for some monotonically increasing function $f:[n]\rightarrow \mathbb{R}$.
Our main result for anonymous utilities is as follows. 
\begin{theorem}\label{theo:ano}
There exists a polynomial algorithm for computing the maximum revenue and an optimal signaling policy for a sender that has a monotone, anonymous utility function.
\end{theorem}
\subsection{Proof of Theorem \ref{theo:ano}}
We show that the concave closure of anonymous function can be computed in polynomial time. 

We use $\mathcal{S}_k$ to denote all the size-$k$ subsets of $[n]$, $\mathcal{S}_k :=\{ S \subseteq [n] \mid |S|=k \}$. We denote by $\mrg(\mu)_i:=\sum_{S\subset [n] : i\in S}\mu(S)$ the marginal probability of the $i$th coordinate to be equal $1$. Note that the constrains of  the concave closure $V^+(a)$ can be written as $\mrg(\mu)_i \leq a_i$ for every $i\in [n]$.

The following lemma from \cite{AB} characterizes the maximum probability mass that  can be assigned to subsets of size $k$ under the constraints imposed by the persuasion levels profile $a= (a_1,...,a_2)$. 

\begin{lemma}[\cite{AB}]\label{lem:AB}
Let $1\geq a_1\geq a_2\geq\ldots\geq a_n\geq 0$ be a monotonic sequence. The solution for the following maximization problem
\begin{align*}
\max & \ \ \sum_{S\in \mathcal{S}_k} \mu(S) \\
\text{ subject to } & \ \ \mrg(\mu)_i \leq a_i \qquad \forall i \in [n] \numberthis \label{eq:maxi-k} \\
& \ \  \mu_S \geq 0 
\end{align*}
is given by
\begin{align*}
\beta_k(a_1,...,a_n) & =\min_{0\leq m<k} \frac{1}{k-m}(a_{m+1}+\ldots+a_n).\\
\end{align*}
Moreover, such a measure $\mu$ that maximizes \eqref{eq:maxi-k} can be computed in polynomial time.
\end{lemma}
The key idea is to use this lemma to solve the LP corresponding to the concave closure---which has exponential (in $n$) number of variables---by another LP that has a polynomial number of variables. We split the original problem into $n$ problems of finding a measure $\mu_k$ over $\mathcal{S}_k$ for every $k=1,...,n$ (the final measure is defined by $\mu=\mu_1+...+\mu_n$). The new maximization problem has $n^2$ variables $(a_i^j)_{i,j\in [n]}$, where $(a_1^j,...,a_n^j)$ represents the marginal constrain vector on subsets of size $k$. We denote by $\alpha_k$ the measure that is assigned to subsets of size $k$, then the original maximization problem can be translated to the following  
\begin{align*}
&\max \text{ } \alpha_1 f(1) + \alpha_2 f(2)+... + \alpha_n f(n) \numberthis \label{eq:polyLP}\\ 
&\text{ s.t. } \sum_{k\in [n]} \alpha_k=1, \text{ } 0\leq \alpha_k \leq \beta_k(a^k_1,a^k_2,...,a^k_n) \text{ for } k\in [n], \text{ and } \sum_{j\in [n]} a^j_i \leq a_i.
\end{align*}
where the first constrain corresponds to $\sum_{S\subset [n]}\mu_S=1$, the second follows from Lemma \ref{lem:AB}, and the last constrain uses the fact that marginals preserve additivity, and thus correspond to $\sum_{S\subset [n]}\mu_S 1_S\leq a$.
 
The following Lemma shows that adding the constrains $a_1^k\geq a_2^k \geq ...\geq a_n^k$ for every $k\in [n]$ to the maximization problem \eqref{eq:polyLP} do not change the maximum value. 

\begin{lemma}\label{lem:mo}
For $(\mu_S)_{S\subset [n]}$ we denote by $a_i^k=\sum_{S\in \mathcal{S}_k : \  i\in S} \mu(S)$. There exists $\mu$ that maximizes \eqref{eq:maxi2} that satisfies $a_1^k\geq a_2^k \geq ...\geq a_n^k$ for every $k\in [n]$.
\end{lemma}

The proof of this Lemma is relegated to Section \ref{sec:lem-proof}.


Note that for the case where $a_1^k\geq a_2^k \geq ...\geq a_n^k$ the constraint $\alpha_k \leq \beta_k(a^k_1,a^k_2,...,a^k_n)$ can be written as
\begin{align*}
\begin{cases}
\alpha_k \leq \frac{1}{k}(a_1+a_2+...+a_n) \\
\alpha_k \leq \frac{1}{k-1}(a_2+a_3...+a_n) \\
\vdots \\
\alpha_k \leq \frac{1}{1}(a_k+a_{k+1}+...+a_n)
\end{cases}
\end{align*}
because $\beta_k$ is defined as the minimum of the right hand side expressions.
Therefore the maximization problem \eqref{eq:polyLP} with the additional constrains $a_1^k\geq a_2^k \geq ...\geq a_n^k$ is an LP maximization with $poly(n)$ number of variables, and it solves the original maximization problem. Moreover, since the proof of Lemma \cite{AB} is constructive, after we have computed the values of $(a_i^j)_{i,j\in [n]}$ that maximize \eqref{eq:polyLP} we can compute the optimal policy as well.

\subsubsection{Proof of Lemma \ref{lem:mo}}\label{sec:lem-proof}
The proof builds upon ideas that were used in the proof of Lemma \ref{lem:AB} in \cite{AB}.
  
Let $\nu$ be a distribution that satisfies the constrains $\mrg(\nu)_i\leq a_i$, and let $\alpha_k=\nu(\mathcal{S}_k)$ be the weight of $\nu$ on subsets of size $k$. It is sufficient to construct another distribution $\mu$ that satisfies $\mu(\mathcal{S}_k)=\alpha_k$ (and thus $\mu$ has the same revenue as $\nu$), and in addition $a_1^k\geq a^k_2 \geq ... \geq a^k_n$. 

The construction is done in $n$ steps, where the steps $k=n,n-1,...,1$ are done in an decreasing order. At step $k$ we assign a measure of $\alpha_{k}$ to subsets of size $k$, and we denote the assigned measure by $\mu_k$. Each step $k$ is done in finite number of stages (at most $n$ stages). Here we describe the assignment of measure at stage $k.m$.

During the construction we "assign mass" and thus, we "spend marginal constrains". We take track of the remaining marginal constrains vector. At the beginning, we set the constrains vector $(a_1^{n.0},...,a_n^{n.0})=(a_1,...,a_n)$ to be the original constrains.

During the process we preserve the monotonicity of the marginal constrains vector and therefore we can denote the marginal constrains vector at stage $k.m$ by
\begin{align*}
(a_1^{k.m},...,a_n^{k.m})=(b_1,...,b_{j},\underbrace{c,c,...,c}_{l-j \text{ times}},b_{l+1},...,b_n)
\end{align*}
where $b_{j}>c>b_{l+1}$ and $j<k\leq l$. Note that if $a^{k.j}_k=a^{k.j}_{k+1}=...=a^{k.j}_n$ then $l=n$ and for simplicity of notation we denote $b_{n+1}=0$. Note that if $a^{k.j}_1=a^{k.j}_{2}=...=a^{k.j}_k$ then $j=0$, and for simplicity of notation we denote $b_0>b_1$. 

At stage $k.m$, the idea is to distribute mass equally over the subsets $S$ of size $k$ that satisfy $[j]\subseteq S \subseteq [l]$ (we have $\binom{l-j}{k-j}$ such sets). If we do so, after we have distributed $x$ units of mass the remaining marginal constrains vector will be 
\begin{align}\label{eq:marg-con1}
b(x)=(b_1-x,...,b_j-x,c-\frac{k-j}{l-j}x,...,c-\frac{k-j}{l-j}x,b_{l+1},...,b_n)
\end{align}
because every element $i=j+1,j+2,...,l$ appears in exactly $\frac{k-j}{l-j}$ fraction of the above subsets. Step $k.m$ terminates at the moment when one of the following three happens:

\begin{itemize}
\item[(1)] The total mass that has been assigned during step $k$ reaches $\alpha_k$. In such a case we proceed to step $k-1$.
\item[(2)] The $j$th coordinate becomes equal to the $(j+1)$th coordinate. In such a case we proceed to stage $k.(m+1)$.
\item[(3)] The $l$th coordinate becomes equal to the $(l+1)$th coordinate. In such a case we proceed to stage $k.(m+1)$.
\end{itemize}

We denote by $\alpha_{k.m}$ the amount of mass that has been assigned during step $k.m$. We denote by $b(\alpha_{k.m})$ the marginal constrains vector after step $k.m$, where $b(\cdot)$ is defined in equation \eqref{eq:marg-con1}. This marginal constrains serves as the marginal constrain vector for the next step (in case (1) happens) or the next stage (in case (2) or (3) happens).
 
We argue the following two statements, which will complete the proof. 
\begin{enumerate}
\item The described process succeeds to complete all the $n$ steps.
\item The described process at each step $k$ assigns mass in a way that $a_1^k\geq a_2^k \geq ... \geq a_n^k$.
\end{enumerate}

Statement (2) follows from the fact that at each stage $k.m$ the marginals of the assigned mass is of the form $(\underbrace{x,...,x}_{j_m \text{ times}},\underbrace{cx,...,cx}_{l_m-j_m \text{ times}},0,...,0)$ for $x=\alpha_{k.m}$ and $c<1$. Moreover, during step $k$ the coordinate $j_m$ is monotonically decreasing, and the coordinate $l_m$ is monotonically increasing. Therefore, the sum of those vectors, which is equal to the vector $(a_1^k, a_2^k, ..., a_n^k)$ is monotonically increasing.

Assume by way of contradiction that statement (1) is false. The above process cannot assign the required measure only if we are at step $k$ and the marginal constrains vector becomes $(d_1,d_2,...,d_m,0,...,0)$ for $m<k$. In such a case indeed the process cannot proceed, because it will turn the $m+1$ coordinate of the marginal constrain vector negative. We denote by $\alpha'_k$ the measure at step $k$ that has been assigned up to the moment of termination. 

We argue that this is impossible from the fact that $\alpha_n,\alpha_{n-1},...,\alpha_k$ are feasible weights for some distribution $\nu$. The idea is that the described above process has minimal marginals on the elements $m+1,...,n$, thus if this process cannot proceed neither could some other distribution $\nu$.  Formally, we denote $\nu=\nu_1+...+\nu_n$, where $\nu_j$ is a measure over $\mathcal{S}_j$. Note that $|\nu_j|=\alpha_j$. We denote $(d_i^j)_{i\in [n]}$ the marginals of $\nu_j$. We argue that $\sum_{i=m+1}^n d_i^j \geq (j-m)\alpha_j$, because every subset of size $j$ contains at least $j-m$ elements from the set $\{m+1,...,n\}$. Therefore we have 
\begin{align}\label{eq:tot-mar}
a_{m+1}+...+a_n \geq \sum_{j=k}^n \sum_{i=m+1}^n d_i^j \geq \sum_{j=k}^n (j-m)\alpha_j
\end{align}

On the other hand, the constructed measure $\mu_n$ with marginals $(a_i^j)_{i\in [n]}$ satisfies $\sum_{i=m+1}^n a_i^j = (j-m)\alpha_j$, because this process assigns positive probability \emph{only} to subsets that contain $\{1,...,m\}$ (because $m<k\leq j$ and $a^j_m>a^j_{m+1}$). Since the process spent all the marginal constrains $a_{m+1},...,a_n$ we have
\begin{align}\label{eq:tot-mar2}
a_{m+1}+...+a_n = \sum_{j=k}^n \sum_{i=m+1}^n a_i^j = (k-m)\alpha'_k + \sum_{j=k+1}^n (j-m)\alpha_j < \sum_{j=k}^n (j-m)\alpha_j
\end{align}
Inequalities \eqref{eq:tot-mar} and \eqref{eq:tot-mar2} yield a contradiction. This completes the proof of the lemma.

\section{Submodular Utilities}
\label{sect:submod}

Recall that a function $V$ is \emph{submodular} if
$V(S\cup\{i\})-V(S)\geq V(T\cup\{i\})-V(T)$ for every $T\subset S \subset [n]$ and every $i\in [n]$. This section considers private Bayesian persuasion settings in which the sender's utility function is submodular. In particular, we develop a tight $(1 - 1/e)$ approximation of 
the optimal signaling policy when the sender's utility is a monotone submodular function. 

It is relevant to note that our algorithmic results require only query access to the submodular function, i.e., our results hold as long as we have access to $V(S)$, for any subset $S \subseteq [n]$. 

We begin by noting that finding the concave closure of a submodular function is {\rm NP}-hard: Given a succinct, monotone, submodular function $f: 2^{[n]} \rightarrow \mathbb{R}$ and a vector $a \in [0,1]^n$, it is {\rm NP}-hard to compute the concave closure $f^+(a)$; see, e.g.~\cite{vondrak2007submodularity,dug}.

 



\subsection{Approximation Algorithm for Submodular Utilities}

This section provides a $(1-1/e)$-approximation algorithm for computing the concave closure of a monotone, submodular function $V$. 
We obtain the $\left( 1 - \frac{1}{e} \right)$ approximation by reducing the computation of the concave closure to the problem of maximizing a submodular function subject to a matroid constraint. The key implication of this approximation result is the following theorem. 

\begin{theorem}
\label{thm:submod}
If in a private Bayesian persuasion problem the utility of the sender, $V$, is a monotone submodular function, then, we can compute  in polynomial time a signaling policy that achieves a revenue of at least $(1-1/e - \varepsilon)$ times the optimal; here, $\varepsilon$ is an arbitrarily small constant. 
\end{theorem}

We present here an outline of the proof of Theorem \ref{thm:submod}, the details appear below in Section \ref{sect:conclose}.

\noindent
\emph{Proof outline of Theorem \ref{thm:submod}} The idea is to reduce the \emph{continuous} maximization problem of the concave closure to a \emph{discrete} optimization problem. One natural way to do so is by rounding the underlying probabilities to integer multiples of a parameter $\delta$ which is equal to $\frac{1}{\text{poly}(n)}$; i.e., to restrict our attention to the case where $\mu_S$ is an integer multiple of $\delta$ for every $S \subset [n]$. An involved part of the proof is to show that by imposing such a restriction we do not loose much in the approximation ratio.

After the reduction to a discrete optimization problem, we show that this new problem corresponds to maximizing a submodular function subject to a matroid constraint (in fact a partition matroid constraint). Finally, we use the result by \cite{ccpv} to conclude that the discrete problem (and, hence, the concave closure) admits a polynomial-time algorithm $(1-\frac{1}{e})$- approximation algorithm.

\subsubsection{Proof of Theorem \ref{thm:submod}}
\label{sect:conclose}

We show that the concave closure of submodular function $V$ at vector $a = (a_1, a_2, \ldots, a_n) \in [0,1]^n$ can be approximated within a factor of $\left(1-\frac{1}{e}-\varepsilon \right)$ for an arbitrarily small $\varepsilon>0$, then Theorem \ref{thm:submod} follows from Lemma \ref{lem:aprox}. 

We split the marginal values $a_i$ into two sets: $\{a_i:a_i\geq \frac{1}{n^2}\}$ are the \emph{high values} and $\{a_i:a_i< \frac{1}{n^2}\}$ are the \emph{low values}. Without loss of generality we assume that $a_1,...,a_m$ are the high values and $a_{m+1},...,a_n$ are the low values, for $m\leq n$.

Every distribution $\mu$ over subsets of $[n]$ induces a distribution $\nu=\nu(\mu)$ over subsets of $[m]$ in the following natural way: the probability mass $\mu_S$ on $S\subset [n]$ is moved to the set $S\cap [m]$,  formally for each subset $T \subseteq [m]$ define $\nu_T :=\sum_{S\subset [n]: S\cap [m]=T}\mu_S$. The following lemma holds for the distribution $\nu=\nu(\mu)$. 

\begin{lemma}\label{lem:sub}
For every distribution $\mu$ that satisfies the marginal constraints (i.e., $ \sum_{S \subseteq [n] : S \ni i } \mu_S \leq a_i$) we have
\begin{align*}
\sum_{S\subset [n]} \mu_S V(S)\geq \sum_{T\subset [m]} \nu(T)V(T)+\sum_{i=m+1}^n a_i V(\{i\}).
\end{align*}
\end{lemma}
\begin{proof}
\begin{align*}
\sum_{S\subset [n]} \mu_S V(S) &\leq \sum_{S\subset [n]} \mu_S \ \left[  V(S\cap [m])+\sum_{i\in S,i>m} V(\{i\})  \right]\\ 
&=\sum_{T\subset [m]} \nu_T V(T)+ \sum_{S} \sum_{i\in S,i>m} \mu_S V(\{i\})\\ 
&= \sum_{T\subset [m]} \nu_T V(T)+ \sum_{i>m} \sum_{S:i\in S} \mu_S V(\{i\}) \\ 
&\leq \sum_{T\subset [m]} \nu_T V(T) + \sum_{i>m} a_i V(\{i\}),
\end{align*}
where the first inequality follows from subadditivity of $V$. The second equation follows from the definition of $\nu=\nu(\mu)$. The third equation is obtained by changing the order of summation and the last inequality follows from the fact that $\mu$ satisfies the marginal constraints.
\end{proof}

We can consider the optimization problem corresponding to the concave closure restricted to the set $[m]$:
\begin{align}
V^+_m(a) & := \max \sum_{T\subset [m]} \nu_T V(T) \label{eq:max-m} \\
& \qquad \text{s.t.}  \  \sum_{T \subseteq [m] : T \ni i } \nu_T \leq a_i \qquad \forall  i \in [m] \nonumber \\
&  \qquad \qquad \nu \textrm{ is a probability measure.} \nonumber 
\end{align}

Given a distribution $\overline{\nu}$ that $\alpha$-approximates problem \eqref{eq:max-m}, we define distribution $\overline{\mu}$ over $[n]$ as follows:
For each subset $T \subseteq [m]$, set $\overline{\mu}_T := (1-\frac{1}{n})\overline{\nu}$. In addition, for every $i > m$ set $\overline{\mu}_{\{i\}} := a_i$. Finally, to ensure that $\overline{\mu}$ is a probability measure we assign a probability mass of $c=\frac{1}{n}-a_{m+1}-...-a_n>0$ to the empty set, i.e., $\mu_\phi :=c$.

\begin{lemma}\label{lem:n->m}
If a distribution $\overline{\nu}$ $\alpha$-approximates problem \eqref{eq:max-m}, then $\overline{\mu}$ provides a $(1-\frac{1}{n})\alpha$-approximation of the original concave closure problem \eqref{eq:maxi2}.
\end{lemma}
\begin{proof}
\begin{align*}
\sum_{S\subset [n]} \overline{\mu}_S V(S) &= \left(1-\frac{1}{n} \right)\sum_{T\subset [m]} \overline{\nu}_T V(T) + \sum_{i>m} a_i V(\{i\})\\
&\geq \left(1-\frac{1}{n} \right)\alpha V^+_m(a) + \sum_{i>m} a_i V(\{i\}) \\
&\geq \left(1-\frac{1}{n} \right)\alpha [ V^+_m(a) + \sum_{i>m} a_i V(\{i\}) ] \\
&\geq \left(1-\frac{1}{n} \right)\alpha V^+(a)
\end{align*}
where the first equation is implied by the definition of $\overline{\mu}$. The second inequality follows from the fact that $\overline{\nu}$ $\alpha$-approximates the concave closure (on the set $[m]$). The third inequality is trivial and the last one follows from Lemma \ref{lem:sub}.  
\end{proof}
Lemma \ref{lem:n->m} reduces the original concave closure problem to the problem of computing the concave closure over $[m]$ where (unlike the original problem) we know that $a_i\geq \frac{1}{n^2}$ for each $i\in [m]$. In the remainder of the proof, we consider the later problem. The idea is to translate this problem into a discrete one. A natural way do to so is by rounding the underlying terms to integer multiples of a parameter $\delta:=\frac{1}{n^4(n+1)}$ and then working with the multiples, instead of the fractional terms. 

Since \eqref{eq:max-m} is a linear program (over variables $ \{\nu_T\}_{T \subseteq [m]}$) with at most $n+1$ non-trivial constraints, without loss of generality we can restrict attention to solutions that have support size of at most $n+1$. 

As mentioned previously, we set a grid of size $\delta:=\frac{1}{n^4(n+1)}$, and we consider the maximization problem of $V^+_m$ where we restrict the probabilities $\{\nu_S\}$ to be integer multiples of $\delta$.

\begin{align}
V^+_m(a) & := \max \sum_{T\subset [m]} \nu_T V(T) \label{eq:max-dis} \\
& \qquad \text{s.t.}  \  \sum_{T \subseteq [m] : T \ni i } \nu_T \leq a_i \qquad \forall  i \in [m] \nonumber \\
&  \qquad \qquad \nu \textrm{ is a probability measure.} \nonumber \\
& \qquad \qquad \nu_T \in \{0, \delta , 2\delta, \ldots,1\}.  \nonumber
\end{align}

\begin{lemma}\label{lem:dis}
Let parameter $ \alpha \in [1/2,1]$. If distribution $\widehat{\nu}$ is an $\alpha$-approximate solution of optimization problem \eqref{eq:max-dis} with support size at most $n+1$, then  $\widehat{\nu}$ is a $(1-\frac{2}{n})\alpha$-approximate solution of the concave closure $V^+_m(a)$ as well. 
\end{lemma}

\begin{proof}
We prove that by restricting attention to probabilities in the set $\{0, \delta , 2\delta, \ldots,1\}$, we incur at most a multiplicative loss of $(1-\frac{1}{n\alpha})$. Given a distribution $\nu$ with support size at most $n+1$ we round down the probabilities to integer multiples of $\delta$ (and put all the remaining probability mass on the empty set), we denote the resulting distribution by $\nu'$. Formally $\nu'_T= \ell \delta$ where $k=\max \{j\in \mathbb{Z}: j \delta \leq \nu_T\}$. Note that
\begin{align}
&\sum_{T} \nu_T V(T) - \sum_{T} \nu'_T V(T)=\sum_{T} (\nu_T-\nu'_T) V(T)
\leq \sum_{T} \delta V(T) \nonumber \\
&\leq \sum_{T\in \Supp(\nu)} \frac{1}{n^4(n+1)} V([m]) \leq \frac{1}{n^4}V([m])\leq \frac{1}{n^4}\sum_{i\in [m]} V(\{i\}) \leq \frac{1}{n^2}\sum_{i\in [m]}a_i V(\{i\}) \label{eq:round}.
\end{align}
Since $\nu_T-\nu'_T \leq \delta$ for each subset $T$, we get the first inequality. To obtain the second inequality we use the definition of $\delta:=  \frac{1}{n^4(n+1)}$ and the monotonicity of $V$. Now, given that the support size of $\nu$ is at most $n+1$, we have the third inequality. The fourth inequality follows from the submodularity (subadditivity) of $V$ and the last inequality follows from the fact that $a_i \geq \frac{1}{n^2}$. 

It is relevant to note that $V^+_m(a)\geq \frac{1}{n}\sum_{i\in [m]} a_i V(\{i\})$, since one feasible solution is to put a probability mass of $\frac{a_i}{n}$ on singleton $\{i\}$ for all $i \in [m]$, and assign the remaining probability mass to the empty set. This is indeed a feasible solution because $\sum_i \frac{a_i}{n} \leq \sum_i \frac{1}{n}\leq 1$.

Finally let $\overline{\nu}$ be an $\alpha$-approximation solution of $V^+_m(a)$, and let $\overline{\nu}'$ be the corresponding rounded-down distribution . Then 
\begin{align*}
\frac{\sum_{T} \overline{\nu}'_T V(T)}{\sum_{T} \overline{\nu}_T V(T)}
&=1-\frac{\sum_{T} \overline{\nu}_T V(T) - \sum_{T} \overline{\nu}'_T V(T)}{\sum_{T} \overline{\nu}_T V(T)}\\
&\geq 1-\frac{\frac{1}{n^2}\sum_{i\in [m]}a_i V(\{i\})}{\sum_{T} \overline{\nu}_T V(T)}\\
&\geq 1-\frac{\frac{1}{n^2}\sum_{i\in [m]}a_i V(\{i\})}{\alpha\frac{1}{n}\sum_{i\in [m]} a_i V(\{i\})} = 1-\frac{1}{n\alpha}
\end{align*}
where the second inequality follows from  bound \eqref{eq:round}, and the third from the fact that $V^+_m(a)\geq \frac{1}{n}\sum_{i\in [m]} a_i V(\{i\})$. Finally, the stated claim follows from the fact that $\alpha \geq 1/2$.
\end{proof}

Lemma \ref{lem:dis} allows us to restrict our attention to the discretized problem \eqref{eq:max-dis}. Note that problem \eqref{eq:max-dis} is in fact equivalent to

\begin{align}
\max_{S^1, \ldots, S^k \subseteq [m]}  & \ \ \ \frac{1}{k} \sum_{j=1}^k V(S^j) \label{prob} \\
\textrm{subject to } & \ \ \ | \{ j \in [k] \mid i \in S^j \} | \leq k_i \qquad \forall i \in [n] \nonumber
\end{align}
where $k_i=  a_i n^4(n+1) $ and $k = n^4(n+1)$. 

Below we show that \eqref{prob} admits a $(1- 1/e)$ approximation. We do so by showing that \eqref{prob} corresponds to the problem of maximizing a monotone submodular function subject to a matroid constraint. Note that a $(1- 1/e)$-approximate solution of \eqref{prob} can be efficiently mapped into a distribution $\hat{\nu}$ that  $(1-\frac{2}{n}) (1-1/e)$-approximates \eqref{eq:max-dis} and has support size $n+1$. This follows from the following observations: given a $(1- 1/e)$-approximate solution of  \eqref{prob},  $\mathcal{S}:= \{S^1, \ldots, S^k \}$, the uniform distribution (say, $\eta$) over $\mathcal{S}$ is a feasible solution of \eqref{eq:max-dis} that achieves an approximation ratio of $(1-1/e)$. We can now consider the following linear program that has the same marginal constraints as \eqref{eq:max-dis} but considers probability distributions whose support lie in $\mathcal{S}$:
\begin{align}
\max & \sum_{T \in \mathcal{S}} \nu_T V(T) \label{eq:interim} \\
\text{s.t.} &  \  \sum_{T \in \mathcal{S} : T \ni i } \nu_T \leq a_i \qquad \forall  i \in [m] \nonumber \\
& \  \sum_{T \in \mathcal{S}} \nu_T = 1 \textrm{ and } \nu_T \geq 0 \quad \forall T \in \mathcal{S} \nonumber 
\end{align}

Linear program \eqref{eq:interim} has $k=n^4(n+1)$ variables $n+1$ non-trivial constraints, hence in polynomial time we can find an optimal solution, say $x:=\{x_T\}_{T \in \mathcal{S}}$, of \eqref{eq:interim} with support size $n+1$. Furthermore, using $\eta$ (the uniform distribution over $\mathcal{S}$) we can show that the optimal value of \eqref{eq:interim} (i.e., $\sum_{T\in \mathcal{S}} x_T V(T) $) is at least $(1- 1/e)$ times the optimal value of \eqref{eq:max-dis}. Hence, by rounding the probabilities $x_T$s down to the nearest multiple of $\delta$, and using arguments similar to the ones presented in Lemma \ref{lem:dis} we can obtain a distribution $\hat{\nu}$ that satisfies the constraints of Lemma~\ref{lem:dis} (i.e., has support size of at most $n+1$) with $\alpha :=(1-\frac{2}{n}) (1-1/e)$.

To complete the proof we show that \eqref{prob} corresponds to the problem of maximizing a monotone submodular function subject to a matroid constraint. Consider base set ${U} = [m] \times [k]$. We get that the size of $U$ is polynomially bounded. 
For a subset  $R = \{ (i_1, j_1), (i_2, j_2), \ldots, (i_l, j_l) \}$ of $U$ and $j \in [k]$, write $R^j$ to denote the projected subset $\{i' \in [m] \mid (i', j) \in R  \}$.

With this notation in hand, define function $F$ for each subset $R = \{ (i_1, j_1), (i_2, j_2), \ldots, (i_l, j_l) \} \subset U$ as follows
\begin{align}
F(R) := \frac{1}{k} \sum_{j=1}^k V(R^j).
\end{align}

We claim that $F$ is submodular: consider subsets $X \subset Y \subset U $ and element $(i,j) \in U$. Note that $F(X + (i,j)) - F(X) = \frac{1}{k} V(X^j + i) - \frac{1}{k} V(X^j)$ and $F(Y+(i,j)) - F(Y) = \frac{1}{k} V(Y^j + i) - \frac{1}{k} V(Y^j)$. Since $X^j \subset Y^j$, the submodulartiy (monotonicity) of $V$ implies the submodularity (monotonicity) of $F$.

Next we consider a partition matroid $\mathcal{M}$ over $U$. Specifically, we say that a subset $R \subset U$ is independent (with respect to the matroid $\mathcal{M}$) iff $| \{ (i',j') \in R \mid i'=i \} | \leq k_i$ for all $i \in [n]$. Note that this is a partition matroid where the disjoint partitions are $B_i := \{(i,1), (i,2),\ldots, (i,k) \}$ and the cardinality bounds are $k_i$s. In other words, we obtain $\mathcal{M}$ by defining $R$ to be an independent subset iff $|R \cap B_i| \leq k_i $ for all $i$.

Note that if a subset $R \subset U$ is independent then $R^1, R^2, \ldots R^k$ satisfy the constraints of the optimization problem (\ref{prob}), i.e., for an independent $R$ we have $|\{ j \in [k] \mid i \in S^j \} | \leq k_i$ for all $i$. 

Overall, we get that optimization problem (\ref{prob}) is equivalent to the following problem:
\begin{align}
\max_{R \subset U} & \ \ \ F(R) \\
\textrm{subject to } & \ \ \ R \in \mathcal{M} \nonumber
\end{align}

Since this is a submodular maximization problem subject to a matroid constraint, it admits a $(1 - \frac{1}{e})$ approximation; see~\cite{ccpv}. This in turn implies that the original problem admits a $(1-\frac{1}{n})(1-\frac{2}{n})^2 (1 - \frac{1}{e})=(1-\frac{1}{e}-O(\frac{1}{n}))$ approximation. We can set parameters such instead of a multiplicative factor of $(1-\frac{1}{n})(1-\frac{2}{ n})^2$ in the approximation we get a term that is arbitrarily close one. Hence, we get the desired result.

\subsection{Hardness of Approximating the Concave Closure}
This section shows that the $\left( 1 - \frac{1}{e} \right)$ approximation guarantee obtained in Section~\ref{sect:conclose} is tight. In particular, applying the machinery developed by~\cite{klmm} leads us to the following theorem. We note that \cite{klmm} establish the hardness of approximating maximum social welfare in combinatorial auctions and similar tools were developed in~\cite{feige} for studying the inapproximability of the domatic number. 

\begin{theorem}
\label{thm:hard}
Given a monotone, submodular fucntion $V:2^{[n]} \rightarrow \mathbb{R}_+$ and vector $a \in [0,1]^n$, for any $\varepsilon >0$, it is {\rm NP}-hard to approximate the concave closure, $V^+(a)$, by a factor better than $\left( 1 - \frac{1}{e} - \varepsilon \right)$. 
\end{theorem}

\begin{proof}[Sketch]
\cite{klmm} study the combinatorial auction problem where $n$ goods have to be partitioned among $m$ agents whose utilities are submodular functions of the goods assigned to them. In this problem, the objective is to maximize social welfare, i.e., the sum of the utilities of the agents. It is shown in~\cite{klmm} that for this problem and any $\varepsilon >0$ there does not exist a polynomial time algorithm that obtains an approximation ratio better than $\left( 1 - \frac{1}{e} - \varepsilon \right)$, unless {\rm P} $=$ {\rm NP}. 

Specifically,~\cite{klmm} start with a \emph{label-cover} problem where it is {\rm NP}-hard to distinguish whether the optimal value, $OPT(L)$, is one or less than a particular constant, $c<1$. From the given label cover problem they construct a combinatorial auction instance, $I$, wherein the maximum social welfare, $OPT(I)$ is greater than a threshold, $\tau$ if the label cover problem admits a solution of value one. Furthermore, if the optimal value of the label cover problem is less than $c$---i.e., $OPT(L)  \leq c$---then it must be the case that $OPT(I) \leq  \left( 1 - \frac{1}{e} - \varepsilon \right) \tau$. This, overall, establishes a $\left( 1 - \frac{1}{e} - \varepsilon \right)$ hardness-of-approximation bound for the combinatorial auction problem. 

Interestingly, in the constructed instance $I$ all of the $m$ agents have the same monotone, submodular utility function, say, $f :2^{[n]} \rightarrow \mathbb{R}_+ $. We claim that  approximating the concave closure of constructed function $f$ at marginal vector $a :=(\frac{1}{m}, \frac{1}{m}, \ldots, \frac{1}{m})$ by a factor better than $\left( 1 - \frac{1}{e} - \varepsilon \right)$ is {\rm NP}-hard. In particular, if $OPT(L) =1$  then $f^+(a) \geq \tau'$ and, moreover, if $OPT(L) \leq c$ then $f^+(a) \leq \left( 1 - \frac{1}{e} - \varepsilon \right) \tau'$; here $\tau'$ is a fixed parameter.

The proof of this claim can be obtained by considering the subsets in the support of an optimal solution, $\mu^*$, of problem \eqref{eq:maxi2} defined for function $f$. Note that the proof given in~\cite{klmm} proceeds by considering the subsets that constitute the partition of goods among agents in $I$, instead we can focus on subsets in the support of $\mu^*$ to obtain the result for the concave closure. In particular, the arguments presented in~\cite{klmm} go through if, instead of cardinalities, we consider measure of sets and expected values of quantities with respected to $\mu^*$.\footnote{For example, in Lemma 5 in~\cite{klmm}, we can redefine sets $N_1^e$ ($N_2^e$) to be the collection of subsets---instead of collection of players---in the support of $\mu^*$ that cover (do not cover) an edge $e$ in the label cover instance. Along these lines, instead of bounding the cardinalities of $N_1^e$ and $N_2^e$ (which are denoted by $n_1^e$ and $n_2^e$ in ~\cite{klmm}), we can bound the measures $\sum_{S \in N_1^e} \mu^*_S$ and $\sum_{S \in N_2^e} \mu^*_S$. The key Lemmas 4 and 5 of~\cite{klmm} have analogous versions in terms of measures. A central step in Lemma 5 is to bound the gap $\Delta_e$ (see page 7 in~\cite{klmm}), for us this step follows via Jensen's inequality.} This, overall, establishes the desired inapproximability result for the concave closure. 
\end{proof}

\section*{Acknowledgements}
The authors thank Uriel Feige for helpful discussions and references.

\bibliographystyle{plain}
\bibliography{pbp}


\end{document}